\documentclass[10pt,doublecolumn]{IEEEtran}
\usepackage{epsf}
\usepackage{graphicx}
\usepackage{amsfonts,amssymb,amsmath}
\usepackage{algorithmic,algorithm}
\usepackage{subfigure}
\usepackage{latexsym}
\usepackage{bm}
\usepackage[usenames]{color}
\usepackage{url}

\usepackage[normalem]{ulem}

\newtheorem{theorem}{Theorem}
\newtheorem{lemma}{Lemma}

\newcommand{\script}[1]{{{\cal{#1} }}}
\newcommand{\expect}[1]{\mathbb{E}\left[#1\right]}

\begin{document}
\title{Power Aware Wireless File Downloading: A Constrained Restless Bandit Approach}
\author{\authorblockN{Xiaohan Wei, \textit{Student Member, IEEE} and Michael J. Neely, \textit{Senior Member, IEEE}}
\thanks{The authors are with the  Electrical Engineering department at the University of Southern California, Los Angeles, CA.}
\thanks{This material is supported in part  by one or more of:  the NSF grants CCF-0747525, 0964479, 1049541, the Network Science Collaborative Technology Alliance sponsored by the U.S. Army Research Laboratory W911NF-09-2-0053.}}

\maketitle

\begin{abstract}
This paper treats power-aware throughput maximization in a multi-user file downloading system. Each user can receive a new file only after its previous file is finished.  The file state processes for each user act as coupled Markov chains that form a generalized restless bandit system. First, an optimal algorithm is derived for the case of one user. The algorithm maximizes throughput subject to an average power constraint. Next, the one-user algorithm is extended to a low complexity heuristic for the multi-user problem. The heuristic uses a simple online index policy and its effectiveness is shown via simulation. For simple 3-user cases where the optimal solution can be computed offline, the heuristic is shown to be near-optimal for a wide range of parameters.
\end{abstract}

\section{Introduction}

Consider a wireless access point, such as a base station or  femto node, that delivers files to
$N$ different wireless users.  The system operates in slotted time with time slots $t \in \{0, 1, 2, \ldots\}$.
Each user can download at most one file at a time.   File sizes are random and complete delivery of a file requires a random number of time slots. A new file request is made by each user at a random
time after it finishes its previous download.
Let $F_n(t) \in \{0,1\}$ represent the binary \emph{file state process} for user $n \in \{1, \ldots, N\}$.  The state $F_n(t)=1$ means that user $n$ is currently active downloading a file, while the state $F_n(t)=0$ means that user $n$ is currently idle.

Idle times are assumed to be independent and geometrically distributed with parameter $\lambda_n$ for each user $n$, so that the average idle time is $1/\lambda_n$.  Active times depend on the random file size and the  transmission decisions that are made.  Every slot $t$, the access point observes which users are active and
decides to serve a subset of at most $M$ users, where $M$ is the maximum number of simultaneous transmissions allowed in the system ($M < N$ is assumed throughout).
The goal is to maximize a weighted sum of throughput subject to a total average power constraint.

The file state processes $F_n(t)$ are coupled controlled Markov chains that form a total state
$(F_1(t), \ldots, F_N(t))$ that can be viewed as a \emph{restless multi-armed
bandit system}.
Such problems are complex due to the inherent curse of dimensionality.

This paper first computes an online optimal algorithm for 1-user systems, i.e., the case $N=1$.  This simple case avoids the curse of dimensionality and provides valuable intuition.  The optimal policy here is nontrivial and uses the theory of Lyapunov optimization for renewal systems \cite{Neely2010}.
 The resulting algorithm makes a greedy transmission decision that affects success probability and power usage. The decision is
 based on a \emph{drift-plus-penalty} index.   Next, the algorithm is extended as a low complexity online heuristic for the $N$-user problem.   The heuristic has the following desirable properties:
\begin{itemize}
\item  Implementation of the $N$-user heuristic is as simple as comparing indices for $N$ different 1-user problems.

\item The $N$-user heuristic is analytically shown to meet the desired average power constraint.

\item  The $N$-user heuristic is shown in simulation to perform well over a wide range of parameters.  Specifically, it is very close to optimal for example 3-user cases where an offline optimal can be computed.
\end{itemize}

Prior work on wireless optimization uses Lyapunov functions to maximize throughput in
cases where the users are assumed to have an infinite amount of data to
send \cite{tass-server-allocation}\cite{stolyar-greedy}\cite{now}\cite{atilla-fairness-ton}\cite{neely-fairness-ton}\cite{Srikant2011}\cite{Huang2013}, or
when data arrives according to a fixed rate process that does not depend on delays in the
network  (which necessitates \emph{dropping} data if the arrival rate vector is outside of the capacity
region) \cite{now}\cite{neely-fairness-ton}.  These models do not consider the interplay between arrivals at the transport layer and file delivery at the network layer.   The current paper captures this interplay through the binary file state processes $F_n(t)$. This
creates a complex problem of
coupled Markov chains.    This problem is fundamental to
file downloading systems.
The modeling and analysis of these systems is a  significant contribution of the current paper.

Markov decision problems (MDPs) can be solved offline via linear programming \cite{puterman}.  This can be prohibitively complex for large dimensional problems.  Low complexity solutions for coupled MDPs are possible in special cases when the
coupling involves only time average constraints \cite{Neely2012}.  Finite horizon coupled MDPs are treated
via integer programming in \cite{weak1} and via a heuristic ``task decomposition'' method in \cite{weak2}.
The problem of the current paper does not fit the framework of \cite{Neely2012}-\cite{weak2} because it includes both time-average constraints (on average power expenditure) and instantaneous constraints which restrict the number of users that can be served on one slot. The latter service restriction is similar to a traditional restless multi-armed bandit (RMAB) system \cite{whittle}.

RMAB problems are generally complex (see P-SPACE hardness results in \cite{pspace}).  A standard low-complexity heuristic for such problems is the \emph{Whittle's index} technique \cite{whittle}.
Low complexity Whittle indexing has been used in RMAB models for wireless
systems \cite{zhao11}\cite{myopic2008}\cite{whittle-shroff}, where simulations demonstrate near
optimal results.  Certain special cases with symmetry are also known to be optimal \cite{zhao11}\cite{myopic2008}.  Unfortunately, not every RMAB problem has a Whittle's index, and such indices, if they exist, are not always easy to compute.  Further, the Whittle's index framework does not consider
additional time average power constraints.  The algorithm developed in the current paper can
be viewed as a Whittle-like indexing scheme that can always be implemented and that incorporates
average power constraints.
It is likely that the techniques of the current paper can be extended to other constrained
RMAB problems.

\section{Single User Scenario}

Consider a file downloading system that consists of only one user that repeatedly downloads files.
Let $F(t) \in \{0,1\}$ be the file state process of the user.  State ``1'' means there is a file in the system that has not completed its download,
and  ``0'' means no file is waiting.
The length of each file is independent and is either exponentially distributed or geometrically distributed (described in more detail below).  Let $\overline{B}$ denote the expected file size in bits. Time is slotted. At each slot in which there is an active file for downloading, the user makes a service decision that affects both the downloading success probability and the power expenditure. After a file is downloaded, the system goes idle (state $0$) and remains in the idle state for a random amount of time that is independent and geometrically distributed with parameter $\lambda>0$.

A transmission decision is made on each slot $t$ in which $F(t)=1$.  The decision affects the number of bits that are sent, the probability these bits are successfully received, and the power usage.
Let $\alpha(t)$ denote the decision variable at slot $t$ and let $\mathcal{A}$ represent the abstract action set with a finite number of elements.  The set $\script{A}$ can represent a collection of modulation and coding
options for each transmission. Assume also that $\mathcal{A}$ contains an idle action denoted as ``0.'' The
decision  $\alpha(t)$ determines the following two values:
\begin{itemize}
  \item The probability of successfully downloading a file $\phi(\alpha(t))$, where $\phi(\cdot)\in[0,1]$ with $\phi(0)=0$.
  \item The power expenditure $p(\alpha(t))$, where $p(\cdot)$ is a nonnegative function with $p(0)=0$.
\end{itemize}
The user chooses $\alpha(t) = 0$ whenever $F(t)=0$.
The user chooses $\alpha(t) \in \script{A}$ for each slot $t$ in which $F(t) = 1$, with the goal of maximizing throughput subject to a time average power constraint.

The problem can be described by a two state Markov decision process with binary state $F(t)$. Given $F(t)=1$, a file is currently in
the system. This file will finish its download at the end of the slot with probability $\phi(\alpha(t))$. Hence, the transition probabilities out of state $1$ are:
\begin{eqnarray}
Pr[F(t+1) = 0 | F(t)=1] &=& \phi(\alpha(t))  \label{eq:trans1} \\
Pr[F(t+1) = 1 | F(t) = 1] &=& 1-\phi(\alpha(t)) \label{eq:trans2}
\end{eqnarray}
Given $F(t)=0$, the system is idle and will transition to the active state in the next slot with probability $\lambda$,
so that:
\begin{eqnarray}
Pr[F(t+1) = 1 | F(t) = 0] &=& \lambda  \label{eq:trans3} \\
Pr[F(t+1) =0 | F(t) = 0] &=& 1- \lambda \label{eq:trans4}
\end{eqnarray}

Define the throughput, measured by bits per slot (not files per slot) as:
\begin{equation}
    \liminf_{T\rightarrow\infty}\frac{1}{T}\sum_{t=0}^{T-1}\overline{B}\phi(\alpha(t)) \nonumber
\end{equation}
The file downloading problem reduces  to the following:
\begin{align}
    \mbox{Maximize:} &~\liminf_{T\rightarrow\infty}\frac{1}{T}\sum_{t=0}^{T-1}\overline{B}\phi(\alpha(t)) \label{eq:prob-1}\\
    \mbox{Subject to:} &~\limsup_{T\rightarrow\infty}\frac{1}{T}\sum_{t=0}^{T-1}p(\alpha(t))\leq \beta \label{eq:prob-2}\\
    & \alpha(t) \in \script{A} \mbox{ $\forall t \in \{0, 1,2, \ldots\}$ such that $F(t) =1$} \label{eq:prob-3} \\
    & \mbox{Transition probabilities satisfy \eqref{eq:trans1}-\eqref{eq:trans4}} \label{eq:prob-4}
\end{align}
where $\beta$ is a positive constant that determines the desired average power constraint.

\subsection{The memoryless file size assumption}

The above model assumes that file completion success on slot $t$ depends only on the
transmission decision $\alpha(t)$, independent of history.  This implicitly assumes that file
length distributions have a \emph{memoryless property}.

This holds when each file $i$ has independent length $B_i$ that is  \emph{exponentially distributed} with mean length $\overline{B}$ bits, so that:
\[ Pr[B_i > x] = e^{-x/\overline{B}} \mbox{ for $x >0$} \]
For example, suppose the \emph{transmission rate} $r(t)$ (in units of bits/slot) and the \emph{transmission success probability} $q(t)$ are given by general functions of $\alpha(t)$:
\begin{eqnarray*}
 r(t) &=& \hat{r}(\alpha(t)) \\
 q(t) &=& \hat{q}(\alpha(t))
 \end{eqnarray*}
 Then the file completion probability $\phi(\alpha(t))$ is the probability that the \emph{residual} amount of bits in the file is less than or equal to $r(t)$, \emph{and} that the transmission of these residual bits is a success.  By the memoryless property of the exponential distribution, the residual file length is distributed the same as the original
 file length.  Thus, the file success probability function is:
 \begin{eqnarray}
\phi(\alpha(t))
  &=& \hat{q}(\alpha(t))Pr[B_i \leq \hat{r}(\alpha(t))] \nonumber \\
  &=& \hat{q}(\alpha(t)) \int_{0}^{\hat{r}(\alpha(t))} \frac{1}{\overline{B}} e^{-x/\overline{B}} dx \label{eq:approx1}
  \end{eqnarray}

Alternatively, history independence holds when each file $i$ consists of a random number $Z_i$ of fixed length packets, where $Z_i$ is geometrically distributed with mean $\overline{Z} = 1/\mu$.   Assume each transmission
sends exactly one packet, but different power levels affect the transmission success probability $q(t) = \hat{q}(\alpha(t))$.   Then:
\begin{equation} \label{eq:approx2}
\phi(\alpha(t)) = \mu\hat{q}(\alpha(t))
\end{equation}

These memoryless file length assumptions ensure that the file state  can be modeled by a simple
binary-valued process $F(t) \in \{0, 1\}$. However, actual file sizes might be neither exponentially distributed nor geometrically
distributed.  One way to treat general distributions is to \emph{approximate} the file sizes as being memoryless by
using a $\phi(\alpha(t))$ function defined by either \eqref{eq:approx1} or \eqref{eq:approx2}, formed by matching the average file size $\overline{B}$ or average number of packets $\overline{Z}$. The \emph{decisions} $\alpha(t)$ are made according to the algorithm below, but the actual event outcomes that arise from these decisions are not memoryless.   A simulation comparison of this approximation is provided in Section \ref{section:sims}, where it is shown to be remarkably accurate (see Fig. \ref{fig:Stupendous5}).

\subsection{Lyapunov optimization}

This subsection develops an online algorithm for problem \eqref{eq:prob-1}-\eqref{eq:prob-4}. First, notice that file state ``$1$'' is recurrent under any decisions for $\alpha(t)$. Denote $t_k$ as the $k$-th time when the system returns to state ``1.''
Define the \emph{renewal frame} as the time period between $t_k$ and $t_{k+1}$.  Define the \emph{frame size}:
\[ T[k] = t_{k+1} - t_k \]
Notice that $T[k]=1$ for any frame $k$ in which the file does not complete its download.  If the file is completed on frame $k$, then $T[k] = 1 + G_k$, where $G_k$ is a geometric random variable with mean
$\expect{G_k} = 1/\lambda$.   Each frame $k$ involves only a single decision $\alpha(t_k)$ that is made at the beginning of the frame.  Thus, the total power used over the duration of frame $k$ is:
\begin{equation}\label{extra1}
\sum_{t=t_{k}}^{t_{k+1}-1}p(\alpha(t)) = p(\alpha(t_k))
\end{equation}
Using a technique similar to that
proposed in \cite{Neely2010}, we treat the time average constraint in \eqref{eq:prob-2} using a virtual queue $Q[k]$ that is updated every frame $k$ by:
\begin{equation}
    Q[k+1]=\max\left\{Q[k]
     + p(\alpha(t_k)) - \beta T[k],~0\right\} \label{eq:q-update}
\end{equation}
with initial condition $Q[0]=0$.
The algorithm is then parameterized by a constant $V\geq0$ which affects a performance tradeoff.  At the beginning of the $k$-th renewal frame, the user observes virtual queue $Q[k]$ and chooses $\alpha(t_k)$ to maximize the following drift-plus-penalty (DPP) ratio \cite{Neely2010}:
\begin{equation}\label{e4}
    \max_{\alpha(t)\in\mathcal{A}} ~~\frac{V\overline{B}\phi(\alpha(t_k))- Q[k]p(\alpha(t_k))}
        {\mathbb{E}[T(t_{k})|\alpha(t_k)]}
\end{equation}
where $\mathbb{E}[T(t_k)|\alpha(t_k)]$ can be easily computed:
\begin{eqnarray*}
\mathbb{E}[T(t_k)|\alpha(t_k)] = 1+\frac{\phi(\alpha(t_k))}{\lambda} \nonumber
\end{eqnarray*}
Thus, \eqref{e4} is equivalent to
\begin{equation}\label{e5}
    \max_{\alpha(t_k)\in\mathcal{A}} ~~\frac{V\overline{B}\phi(\alpha(t_k))- Q[k]p(\alpha(t_k))}
        {1+\phi(\alpha(t_k))/\lambda}
\end{equation}
Since there are only a finite number of elements in $\mathcal{A}$, (\ref{e5}) is easily computed.  This gives the following algorithm for the single-user case:
\begin{itemize}
  \item At each time $t_{k}$, the user observes virtual queue $Q[k]$ and chooses $\alpha(t_k)$ as the solution to (\ref{e5}) (where ties are broken arbitrarily).
  \item The value $Q[k+1]$ is computed according to \eqref{eq:q-update} at the end of the $k$-th frame.
\end{itemize}

\subsection{Average power constraints via queue bounds}

\begin{lemma} \label{lem:1}
If there is a constant $C\geq 0$ such that $Q[k] \leq C$ for all $k \in \{0, 1, 2, \ldots\}$, then:
\[ \limsup_{T\rightarrow\infty}\frac{1}{T}\sum_{t=0}^{T-1}p(\alpha(t))\leq \beta \]
\end{lemma}
\begin{proof}
From \eqref{eq:q-update}, we know that for each frame $k$:
\begin{equation}
    Q[k+1] \geq Q[k]
     +p(\alpha(t_{k}))-T[k]\beta \nonumber
\end{equation}
Rearranging terms and using $T[k] = t_{k+1} -t_k$ gives:
\[ p(\alpha(t_k)) \leq (t_{k+1} - t_k)\beta + Q[k+1]-Q[k] \]
Fix $K>0$. Summing over $k\in\{0,1,\cdots,K-1\}$  gives:
\begin{eqnarray*}
    \sum_{k=0}^{K-1}p(\alpha(t_{k})) &\leq& (t_{K}-t_0) \beta + Q[K] - Q[0] \\
    &\leq&  t_K\beta + C
\end{eqnarray*}
The sum power over the first $K$ frames is the same as the sum up to time $t_{K}-1$, and so:
\[ \sum_{t=0}^{t_K-1} p(\alpha(t)) \leq t_K \beta + C \]
Dividing by $t_K$ gives:
\[ \frac{1}{t_K}\sum_{t=0}^{t_K-1} p(\alpha(t)) \leq \beta + C/t_K. \]
Taking $K\rightarrow\infty$, then,
\[\limsup_{K\rightarrow\infty}\frac{1}{t_K}\sum_{t=0}^{t_K-1} p(\alpha(t)) \leq \beta, \]
which yields the result by combining \eqref{extra1}.
\end{proof}

The next lemma shows that the queue process under our proposed algorithm is deterministically bounded.
Define:
\begin{eqnarray*}
p^{min} &=& \min_{\alpha\in\mathcal{A}\setminus\{0\}}p(\alpha)\\
p^{max}&=&\max_{\alpha\in\mathcal{A}\setminus\{0\}}p(\alpha)
\end{eqnarray*}

\begin{lemma} \label{lem:2} If $Q[0]=0$, then under our algorithm we have for all $k>0$:
\[ Q[k]\leq \max\left\{\frac{V\overline{B}}{p^{min}}+p^{max}-\beta,0\right\}\]
\end{lemma}

\begin{proof}
First, consider the case when $p^{max} \leq \beta$. From \eqref{eq:q-update} and the fact that $T[k] \geq 1$ for all $k$,
 it is clear  the queue can never increase, and so $Q[k] \leq Q[0]=0$ for all $k>0$.

Next, consider the case when $p^{max} > \beta$. We prove the assertion by induction on $k$.
The result trivially holds for $k=0$.  Suppose it holds at $k=l$ for $l>0$, so that:
\[ Q[l]\leq \frac{V\overline{B}}{p^{min}}+p^{max}-\beta \]
We are going to prove that the same holds for $k=l+1$.  There are two cases:
\begin{enumerate}
  \item $Q[l]\leq\frac{V\overline{B}}{p^{min}}$. In this case we have by \eqref{eq:q-update}:
  \begin{eqnarray*}
   Q[l+1] &\leq& Q[l] + p^{max} - \beta \\
   &\leq& \frac{V\overline{B}}{p^{min}} + p^{max} - \beta
   \end{eqnarray*}

    \item $\frac{V\overline{B}}{p^{min}}<Q[l]\leq\frac{V\overline{B}}{p^{min}}+p^{max}-\beta$.     In this case,
    if $p(\alpha(t_l))=0$ then the queue cannot increase, so:
    \[ Q[l+1] \leq Q[l] \leq   \frac{V\overline{B}}{p^{min}} + p^{max} - \beta  \]
    On the other hand, if $p(\alpha(t_l))>0$ then $p(\alpha(t_l)) \geq p^{min}$ and so the numerator in
    \eqref{e5} satisfies:
    \begin{eqnarray*}
    V\overline{B}\phi(\alpha(t_l)) - Q[l]p(\alpha(t_l))
    &\leq& V\overline{B} - Q[l]p^{min} \\
    &<& 0
    \end{eqnarray*}
    and so the maximizing ratio in \eqref{e5} is negative.  However, the maximizing ratio in \eqref{e5} \emph{cannot} be negative,
    because the alternative  choice $\alpha(t_l)=0$ would increase the ratio to 0.  This contradiction implies that
    we cannot have $p(\alpha(t_l))>0$.
\end{enumerate}
\end{proof}

The above is a \emph{sample path result} that
used only the fact that $\lambda >0$ and $0 < p^{min}\leq p(t) \leq p^{max}$.  Thus, the algorithm meets the average power constraint even if the
$\lambda$,  $\overline{B}$, and $\phi(\alpha(t))$ values used in the algorithm
are only \emph{estimates} of the true values.

\subsection{Optimality over randomized algorithms}

Consider the following class of \emph{i.i.d. randomized algorithms}:  Let $\theta(\alpha)$ be non-negative numbers defined for each $\alpha \in \script{A}$, and suppose they satisfy $\sum_{\alpha \in \script{A}} \theta(\alpha) = 1$.  Let $\alpha^*(t)$ represent a policy that, every slot $t$ for which $F(t)=1$, chooses $\alpha^*(t) \in \script{A}$ by independently selecting strategy $\alpha$ with probability $\theta(\alpha)$.    Then $(p(\alpha^*(t_k)), \phi(\alpha^*(t_k)))$ are independent and identically distributed (i.i.d.) over frames $k$. Under this algorithm, it follows by the law of large numbers that the throughput and power expenditure satisfy (with probability 1):
\begin{eqnarray*}
\lim_{t\rightarrow\infty} \frac{1}{T}\sum_{t=0}^{T-1} \overline{B}\phi(\alpha^*(t)) &=& \frac{\overline{B}\expect{\phi(\alpha^*(t_k))}}{1 + \expect{\phi(\alpha^*(t_k))}/\lambda} \\
\lim_{t\rightarrow\infty} \frac{1}{T}\sum_{t=0}^{T-1} p(\alpha^*(t)) &=& \frac{\expect{p(\alpha^*(t_k))}}{1 + \expect{\phi(\alpha^*(t_k))}/\lambda}
\end{eqnarray*}
It can be shown that optimality of problem \eqref{eq:prob-1}-\eqref{eq:prob-4} can be achieved over this class. Thus, there exists an i.i.d. randomized algorithm $\alpha^*(t)$ that satisfies:
\begin{eqnarray}
\frac{\overline{B}\expect{\phi(\alpha^*(t_k))}}{1 + \expect{\phi(\alpha^*(t_k))}/\lambda} &=& \mu^* \label{eq:iid1} \\
 \frac{\expect{p(\alpha^*(t_k))}}{1 + \expect{\phi(\alpha^*(t_k))}/\lambda} &\leq& \beta \label{eq:iid2}
\end{eqnarray}
where $\mu^*$ is the optimal throughput for the problem \eqref{eq:prob-1}-\eqref{eq:prob-4}.

\subsection{Key feature of the drift-plus-penalty ratio}

Define $\script{H}[k]$ as the \emph{system history} up to frame $k$, which includes all
random events that occurred before frame $k$, and also includes the queue value $Q[k]$ (since
this is determined by the random events before frame $k$). Consider the algorithm
that, on frame $k$, observes $Q[k]$ and
chooses $\alpha(t_k)$ according to \eqref{e5}.  The following key feature of this algorithm can be
shown (see \cite{Neely2010} for related results):
\begin{align*}
&\frac{\expect{-V\overline{B}\phi(\alpha(t_k)) + Q[k]p(\alpha(t_k))|\script{H}[k]}}{\expect{1 + \phi(\alpha(t_k))/\lambda|\script{H}[k]}}
\nonumber \\
&\leq
\frac{\expect{-V\overline{B}\phi(\alpha^*(t_k)) + Q[k]p(\alpha^*(t_k))|\script{H}[k]}}{\expect{1 + \phi(\alpha^*(t_k))/\lambda|\script{H}[k]}}
\end{align*}
where $\alpha^*(t_k)$ is any (possibly randomized) alternative decision that is based only on $\script{H}[k]$.  Using
the i.i.d. decision $\alpha^*(t_k)$ from \eqref{eq:iid1}-\eqref{eq:iid2} in the above and noting that this alternative decision
is independent of $\script{H}[k]$ gives:
\begin{align}
&\frac{\expect{-V\overline{B}\phi(\alpha(t_k)) + Q[k]p(\alpha(t_k))|\script{H}[k]}}{\expect{1 + \phi(\alpha(t_k))/\lambda|\script{H}[k]}}
\leq -V\mu^* + Q[k]\beta \label{eq:key-feature}
\end{align}

\subsection{Performance theorem}

\begin{theorem}
The proposed algorithm achieves the constraint $\limsup_{T\rightarrow\infty}\frac{1}{T}\sum_{t=0}^{T-1}p(\alpha(t))\leq \beta$ and yields throughput satisfying (with probability 1):
\begin{equation}\label{e6}
    \liminf_{T\rightarrow\infty}\frac{1}{T}\sum_{t=0}^{T-1}\phi(\alpha(t))\geq \mu^{*}-\frac{C_0}{V}
\end{equation}
where $C_0$ is a constant.
\end{theorem}
\begin{proof}
First, for any fixed $V$, Lemma \ref{lem:2} implies that the queue is deterministically bounded.  Thus, according to Lemma \ref{lem:1}, the proposed algorithm achieves the constraint $\limsup_{T\rightarrow\infty}\frac{1}{T}\sum_{t=0}^{T-1}\mathbb{E}[p(\alpha(t))]\leq \beta$. The rest is devoted to proving the throughput guarantee \eqref{e6}.

Define:
\begin{equation}
    L(Q[k])=\frac{1}{2}Q[k]^2. \nonumber
\end{equation}
We call this a \emph{Lyapunov function}. Define a frame-based Lyapunov Drift as:
\[ \Delta[k] = L(Q[k+1]) - L(Q[k]) \]
According to \eqref{eq:q-update} we get
\begin{equation}
    Q[k+1]^2\leq \left(Q[k]+p(\alpha(t_{k}))-T[k]\beta\right)^2. \nonumber
\end{equation}
Thus:
\begin{eqnarray*}
\Delta[k]  \leq \frac{(p(\alpha(t_k)) - T[k]\beta)^2}{2} + Q[k](p(\alpha(t_k)) - T[k]\beta)
\end{eqnarray*}
Taking a conditional expectation of the above given $\script{H}[k]$ and recalling that $\script{H}[k]$ includes the information
$Q[k]$ gives:
\begin{equation} \label{eq:drift}
 \expect{\Delta[k] |\script{H}[k]} \leq C_0 + Q[k]\expect{p(\alpha(t_k)) - \beta T[k]| \script{H}[k]}
 \end{equation}
where $C_0$ is a constant that satisfies the following for all possible histories $\script{H}[k]$:
\[ \expect{\frac{(p(\alpha(t_k)) - T[k]\beta)^2}{2} \left |\right. \script{H}[k]} \leq C_0 \]
Such a constant $C_0$ exists because the power $p(\alpha(t_k))$ is deterministically bounded, and the frame sizes $T[k]$ are bounded in second moment regardless of history.

Adding the ``penalty'' $-\expect{V\overline{B}\phi(\alpha(t_k))|\script{H}[k]}$ to both sides of \eqref{eq:drift} gives:
\begin{align*}
&\expect{\Delta[k] - V\overline{B}\phi(\alpha(t_k))|\script{H}[k]}  \\
&\leq C_0 +  \expect{-V\overline{B}\phi(\alpha(t_k)) + Q[k](p(\alpha(t_k)) - T[k]\beta)| \script{H}[k]} \\
&= C_0 - Q[k]\beta\expect{T[k]|\script{H}[k]}  \\
& +  \frac{\expect{T[k]|\script{H}[k]}\expect{-V\overline{B}\phi(\alpha(t_k)) + Q[k]p(\alpha(t_k))|\script{H}[k]}}{\expect{T[k]|\script{H}[k]}}
\end{align*}
Expanding $T[k]$ in the denominator of the last term  gives:
\begin{align*}
&\expect{\Delta[k] - V\overline{B}\phi(\alpha(t_k))|\script{H}[k]}  \nonumber\\
&\leq C_0 - Q[k]\beta\expect{T[k]|\script{H}[k]}  + \expect{T[k]|\script{H}[k]} \times \nonumber \\
&\frac{\expect{-V\overline{B}\phi(\alpha(t_k)) + Q[k]p(\alpha(t_k))|\script{H}[k]}}{\expect{1 + \phi(\alpha(t_k))/\lambda|\script{H}[k]}}
\end{align*}
Substituting \eqref{eq:key-feature} into the above expression gives:
\begin{eqnarray}
&&\hspace{-.4in}\expect{\Delta[k] - V\overline{B}\phi(\alpha(t_k))|\script{H}[k]} \nonumber  \\
&\leq&  C_0 - Q[k]\beta\expect{T[k]|\script{H}[k]} \nonumber  \\
&& + \expect{T[k]|\script{H}[k]} (-V\mu^* + \beta Q[k])\nonumber \\
&=& C_0 -V\mu^* \expect{T[k] | \script{H}[k]}  \label{eq:dude}
\end{eqnarray}
Rearranging gives:
\begin{equation} \label{eq:dpp}
\expect{\Delta[k] + V(\mu^*T[k] - \overline{B}\phi(\alpha(t_k)))|\script{H}[k]} \leq C_0 \
\end{equation}

The above is a drift-plus-penalty expression.
Because we already know the queue $Q[k]$ is deterministically bounded, it follows that:
\[ \sum_{k=1}^{\infty} \frac{\expect{\Delta[k]^2}}{k^2} < \infty \]
Thus, the drift-plus-penalty result in
Proposition 2 of \cite{JAM2012} ensures that (with probability 1):
\[  \limsup_{K\rightarrow\infty} \frac{1}{K}\sum_{k=0}^{K-1} \left[\mu^* T[k] - \overline{B}\phi(\alpha(t_k))\right] \leq \frac{C_0}{V} \]
Thus, for any $\epsilon>0$ one has for all sufficiently large $K$:
\[ \frac{1}{K}\sum_{k=0}^{K-1}[\mu^* T[k] - \overline{B} \phi(\alpha(t_k))] \leq \frac{C_0}{V} + \epsilon \]
Rearranging implies that for all sufficiently large $K$:
\begin{eqnarray*}
\frac{\sum_{k=0}^{K-1} \overline{B}\phi(\alpha(t_k))}{\sum_{k=0}^{K-1} T[k]} &\geq& \mu^*  - \frac{(C_0/V + \epsilon)}{\frac{1}{K}\sum_{k=0}^{K-1}T[k]}\\
&\geq& \mu^* - (C_0/V + \epsilon)
\end{eqnarray*}
where the final inequality holds because $T[k] \geq 1$ for all $k$. Thus:
\[ \liminf_{K\rightarrow\infty} \frac{\sum_{k=0}^{K-1} \overline{B}\phi(\alpha(t_k))}{\sum_{k=0}^{K-1} T[k]} \geq \mu^* - (C_0/V + \epsilon) \]
The above holds for all $\epsilon>0$.  Taking a limit as $\epsilon\rightarrow 0$ implies:
\[ \liminf_{K\rightarrow\infty} \frac{\sum_{k=0}^{K-1} \overline{B}\phi(\alpha(t_k))}{\sum_{k=0}^{K-1} T[k]} \geq \mu^* - C_0/V, \]
which yields the result by noticing that $\phi(\alpha(t))$ only changes at the boundary of each frame.
\end{proof}

The theorem shows that throughput can be pushed within $O(1/V)$ of the optimal value $\mu^*$, where $V$ can be chosen as large as desired to ensure throughput is arbitrarily close to optimal.  The tradeoff is a queue bound that grows linearly with $V$ according to Lemma \ref{lem:2}, which affects the convergence time required for the constraints to be close to the desired time averages (as described in the proof of Lemma \ref{lem:1}).

\section{Multi-User file downloading}

This section considers a multi-user file downloading system that
consists of $N$ single user subsystems.  Each subsystem is similar to the single-user system described in the previous section.
Specifically, for  the $n$-th user (where $n \in \{1, \ldots, N\}$):
\begin{itemize}
\item The file state process is $F_n(t) \in \{0,1\}$.
\item The transmission decision is $\alpha_n(t) \in \script{A}_n$, where $\script{A}_n$ is an abstract set of transmission options for user $n$.
\item The power expenditure on slot $t$ is $p_n(\alpha_n(t))$.
\item The success probability on a slot $t$ for which $F_n(t)=1$ is $\phi_n(\alpha_n(t))$, where $\phi_n(\cdot)$ is the function that describes file completion probability for user $n$.
\item The idle period parameter is $\lambda_n>0$.
\item The average file size is $\overline{B}_n$ bits.
\end{itemize}
Assume that the random variables associated with different subsystems are mutually independent.

To control the downloading process, there is a
central server with only $M$ threads ($M<N$), meaning that \emph{at most $M$ jobs can be processed simultaneously}. So at each time slot, the server has to make decisions selecting at most $M$ out of $N$ users to transmit a portion of their files. These decisions are further restricted by a global time average power constraint. The goal is to maximize the aggregate throughput, which is defined as
\begin{equation}
    \liminf_{T\rightarrow\infty}\frac{1}{T}\sum_{t=0}^{T-1}\sum_{n=1}^Nc_{n}\overline{B}_n\phi(\alpha_{n}(t)) \nonumber
\end{equation}
where $c_1, c_2, \ldots, c_N$ are a collection of positive weights that can be used to prioritize users.
Thus, this multi-user file downloading problem reduces down to the following:
\begin{align}
   \mbox{Max:} &~\liminf_{T\rightarrow\infty}\frac{1}{T}\sum_{t=0}^{T-1}\sum_{n=1}^N
    c_{n}\overline{B}_n\phi_n(\alpha_{n}(t)) \label{eq:multi-1} \\
    \mbox{S.t.:}&~\limsup_{T\rightarrow\infty}\frac{1}{T}\sum_{t=0}^{T-1}\sum_{n=1}^Np_n(\alpha_{n}(t))\leq \beta \label{eq:multi-2}\\
    &~\sum_{n=1}^NI(\alpha_{n}(t))\leq M~~\forall t\in\{0,1,2,\cdots\} \label{eq:mutli-3} \\
    &~Pr[F_n(t+1)=1~|~F_n(t)=0]=\lambda_n\label{eq:multi-4}\\
    &~Pr[F_n(t+1)=0~|~F_n(t)=1]=\phi_n(\alpha_n(t))\label{eq:multi-5}
\end{align}
where the constraints \eqref{eq:multi-4}-\eqref{eq:multi-5} hold for all $n \in \{1, \ldots, N\}$ and $t \in \{0, 1, 2,\ldots\}$, and
where $I(\cdot)$ is the indicator function defined as:
\begin{equation}
    I(x)=\left\{
           \begin{array}{ll}
             0, & \hbox{if $x=0$;} \\
             1, & \hbox{otherwise.}
           \end{array}
         \right.\nonumber
\end{equation}

\subsection{Lyapunov Indexing Algorithm}
This section develops our indexing algorithm for the multi-user case using the single-user algorithm as a stepping stone. The major difficulty is  the instantaneous constraint $\sum_{n=1}^NI(\alpha_{n}(t))\leq M$.  Temporarily neglecting this constraint,
we use Lyapunov optimization to deal with the time average power constraint first.

We introduce a virtual queue $Q(t)$, which is again 0 at $t=0$. Instead of updating it on a frame basis, the server updates this queue every slot as follows:
\begin{equation}\label{m2}
    Q(t+1)=\max\left\{Q(t)+\sum_{n=1}^Np_n(\alpha_n(t))-\beta,0\right\}.
\end{equation}
Define $\script{N}(t)$ as the set of users beginning their renewal frames at time $t$, so that $F_n(t)=1$ for all such users.
 In general, $\script{N}(t)$ is a subset of $\script{N}=\{1,2,\cdots,N\}$. Define $|\script{N}(t)|$ as the number of users
 in the set $\script{N}(t)$.

At each time slot $t$, the server observes the queue state $Q(t)$ and chooses $(\alpha_1(t), \ldots, \alpha_N(t))$
to maximize the following drift-plus-penalty expression subject to an instantaneous constraint:
\begin{eqnarray}
    \mbox{Max.:} & \sum_{n\in\mathcal{N}(t)}\frac{Vc_n\overline{B}_n\phi_n(\alpha_{n}(t))- Q(t)p_n(\alpha_{n}(t))}
        {1 + \phi_n(\alpha_n(t))/\lambda_n} \label{eq:multi-alg1} \\
    \mbox{S.t.:}  & \alpha_n(t) \in \script{A}_n \: \: \forall n \in \script{N} \label{eq:multi-alg2} \\
 &   \alpha_n(t) = 0 \: \: \forall n \notin \script{N}(t) \label{eq:multi-alg3} \\
    &\sum_{n\in\mathcal{N}(t)}I(\alpha_{n}(t))\leq M. \label{eq:multi-alg4}
\end{eqnarray}

Notice that in \eqref{eq:multi-alg1}, the term
\begin{equation}
    g_n(\alpha_n(t))\triangleq\frac{Vc_n\overline{B}_n\phi_n(\alpha_{n}(t))- Q(t)p_n(\alpha_{n}(t))}
        {1+\phi_n(\alpha_n(t))/\lambda_n} \nonumber
\end{equation}
is similar to the expression \eqref{e5} used in the single-user optimization.   Call $g_n(\alpha_n(t))$ a \emph{reward}.
Now define an index for each subsystem $n$ by: \begin{equation}\label{m5}
    \gamma_n(t)\triangleq\max_{\alpha_n(t)\in\mathcal{A}_n}g_n(\alpha_n(t))
\end{equation}
which is the maximum possible reward one can get from the $n$-th subsystem at time slot $t$. Thus, it is natural to define the following myopic algorithm:  Find the (at most) $M$ subsystems in $\script{N}(t)$ with the greatest rewards, and serve these with their corresponding
optimal $\alpha_n(t)$ options in $\script{A}_n$ that maximize $g_n(\alpha_n(t))$.
Specifically:
\begin{itemize}
  \item At each time slot $t$, the server observes virtual queue state $Q(t)$ and computes the indices using (\ref{m5}) for all $n\in\mathcal{N}(t)$.
  \item Activate the $\min[M, |\script{N}(t)|]$ subsystems with greatest indices, using their corresponding actions $\alpha_n(t) \in \script{A}_n$ that
  maximize $g_n(\alpha_n(t))$.
  \item Update $Q(t)$ according to \eqref{m2}  at the end of each slot $t$.
  \end{itemize}

\subsection{Theoretical Performance Analysis}

In this subsection, we show that the above algorithm always satisfies the desired time average power constraint.  Define:
\begin{eqnarray*}
p^{min}_n &=& \min_{\alpha_n\in\mathcal{A}_n\setminus\{0\}}p_n(\alpha_n) \\
p^{min} &=& \min_np^{min}_n \\
p^{max}_n &=& \max_{\alpha_n\in\mathcal{A}_n}p_n(\alpha_n) \\
c^{max} &=& \max_{n}c_n \\
\overline{B}^{max} &=& \max_n \overline{B}_n
\end{eqnarray*}

\begin{lemma} \label{lem:3}
Under the above Lyapunov indexing algorithm, the
queue $\{Q(t)\}_{t=0}^{\infty}$  is deterministically bounded. Specifically, we have for all $t \in \{0, 1, 2, \ldots\}$:
\[ Q(t)\leq \max\left\{\frac{Vc^{max}\overline{B}^{max}}{p^{min}}+\sum_{n=1}^Np^{max}_n-\beta,0\right\}\]
\end{lemma}
\begin{IEEEproof}
First, consider the case when $\sum_{n=1}^Np^{max}_n \leq \beta$. Since $Q(0)=0$, it is clear from the updating
rule \eqref{m2} that $Q(t)$ will remain 0 for all $t$.

Next, consider the case when $\sum_{n=1}^Np^{max}_n > \beta$. We prove the assertion by induction on $t$. The result
trivially holds for $t=0$. Suppose at $t=t'$, we have:
\[ Q(t')<\frac{Vc^{max}\overline{B}^{max}}{p^{min}}+\sum_{n=1}^Np^{max}_n-\beta \]
We are going to prove that the same statement holds for $t=t'+1$. We further divide it into two cases:
\begin{enumerate}
  \item $Q(t')\leq\frac{Vc^{max}\overline{B}^{max}}{p^{min}}$. In this case, since the queue increases by at most
  $\sum_{n=1}^Np^{max}_n - \beta$ on one slot, we have:
  \[ Q(t'+1) \leq \frac{Vc^{max}\overline{B}^{max}}{p^{min}} + \sum_{n=1}^Np^{max}_n - \beta\]

  \item $\frac{Vc^{max}\overline{B}^{max}}{p^{min}}<Q(t')\leq\frac{Vc^{max}\overline{B}^{max}}{p^{min}}+\sum_{n=1}^Np^{max}_n-\beta$. In this case, since $\phi_n(\alpha_n(t'))\leq 1$, there is no possibility that $Vc_n\overline{B}^{max}\phi_n(\alpha_n(t'))\geq Q(t')p_n(\alpha_n(t'))$ and thus $\alpha_n(t')$ must be 0 for all $n$. Thus, all indices are 0. This implies that $Q(t'+1)$ cannot increase, and we get $Q(t'+1)\leq\frac{Vc^{max}\overline{B}^{max}}{p^{min}}+\sum_{n=1}^Np^{max}_n-\beta$.
\end{enumerate}
\end{IEEEproof}

\begin{theorem}
The proposed Lyapunov indexing algorithm achieves the constraint:
\[ \limsup_{T\rightarrow\infty}\frac{1}{T}\sum_{t=0}^{T-1}\sum_{n=1}^Np_n(\alpha_{n}(t))\leq \beta \]
\end{theorem}
\begin{IEEEproof}
Using Lemma \ref{lem:1} under the special case that each frame only occupies one slot, we get that if $\{Q(t)\}_{t=0}^\infty$ is deterministically bounded, then the time average constraint is satisfied. Then, according to Lemma \ref{lem:3} we are done.
\end{IEEEproof}


\section{Simulation Experiments} \label{section:sims}

In this section, we demonstrate the near optimality of the multi-user Lyapunov indexing algorithm by extensive simulations. In the first part, we simulate the case in which the file length distribution is geometric, and show that the suboptimality gap is extremely small. In the second part, we test the robustness of our algorithm for more general scenarios in which the file length distribution is not geometric.
For simplicity, it is assumed throughout that all transmissions send a fixed sized packet, all files are an integer number of these packets, and that decisions $\alpha_n(t) \in \script{A}_n$ affect the success probability of the transmission as well as the power expenditure.

\subsection{Lyapunov Indexing for multi-user downloading with geometric file length} \label{subsection:geometric-sim}
In the first simulation we use $N=3$, $M=1$ with action set $\mathcal{A}_n=\{0,1\}~\forall n$; idle period parameter: $\lambda_1=0.8$, $\lambda_2=0.5$, $\lambda_3=0.1$.  Files consist of an integer number of packets and have independent and geometrically
distributed sizes with parameters $\mu_1=0.1$, $\mu_2=0.2$, and $\mu_3=0.4$; so that the expected file size for user $n \in \{1, 2, 3\}$
is $\overline{B}_n = 1/\mu_n$ packets.  The success probability functions are given by: $\phi_1(1)=0.9\mu_1$, $\phi_2(1)=0.8\mu_2$, $\phi_3(1)=0.7\mu_3$; power expenditure function: $p_1(1)=2$, $p_2(1)=1.5$, $p_3(1)=1$; weight parameters: $c_1=1$, $c_2=1.5$, $c_3=2$ and $\beta=1$. The algorithm is run for 1 million slots. We compare the performance of our algorithm with the optimal randomized policy. The optimal policy is computed by constructing composite states (i.e. if queue 1 is at state 0, queue 2 is at state 1 and queue 3 is at state 1, we view 011 as a composite state, and then reformulating this MDP into a linear program (see \cite{MDP}) which contains 20 variables.

In Fig. \ref{fig:Stupendous2}, we show that as our tradeoff parameter $V$ gets larger, the objective value approaches the optimal value and achieves a near optimal performance. Fig. \ref{fig:Stupendous3} and Fig. \ref{fig:Stupendous4} show that $V$ also affects the virtual queue size and the constraint gap. As $V$ gets larger, the average virtual queue size becomes larger and the gap becomes smaller. We also plot the upper bound of queue size we derived from Lemma \ref{lem:3} in Fig. \ref{fig:Stupendous3}, demonstrating that the queue is bounded.

\begin{figure}[htbp]
   \centering
   \includegraphics[height=2.5in]{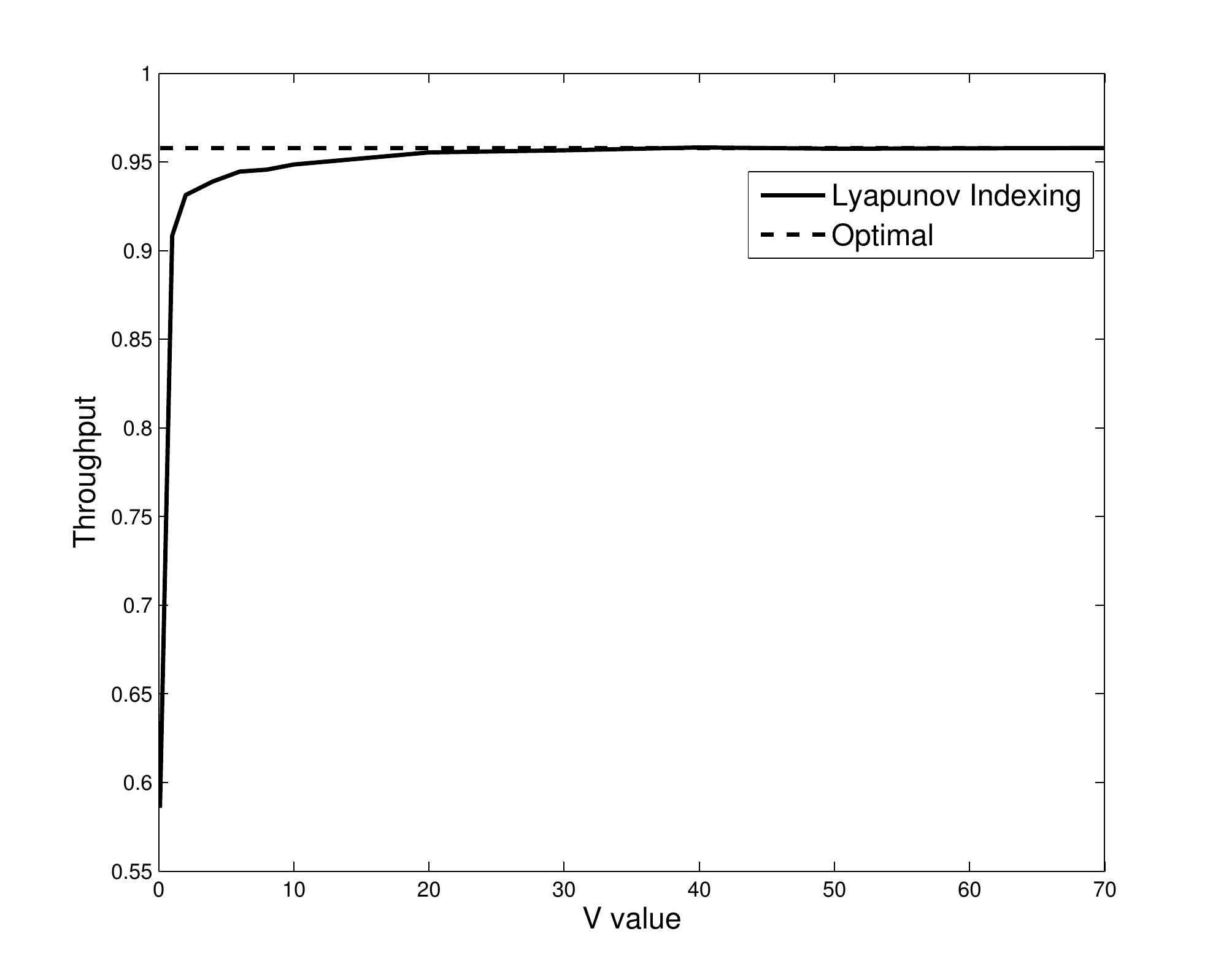} 
   \caption{Throughput versus tradeoff parameter V}
   \label{fig:Stupendous2}
\end{figure}

\begin{figure}[htbp]
   \centering
   \includegraphics[height=2.5in]{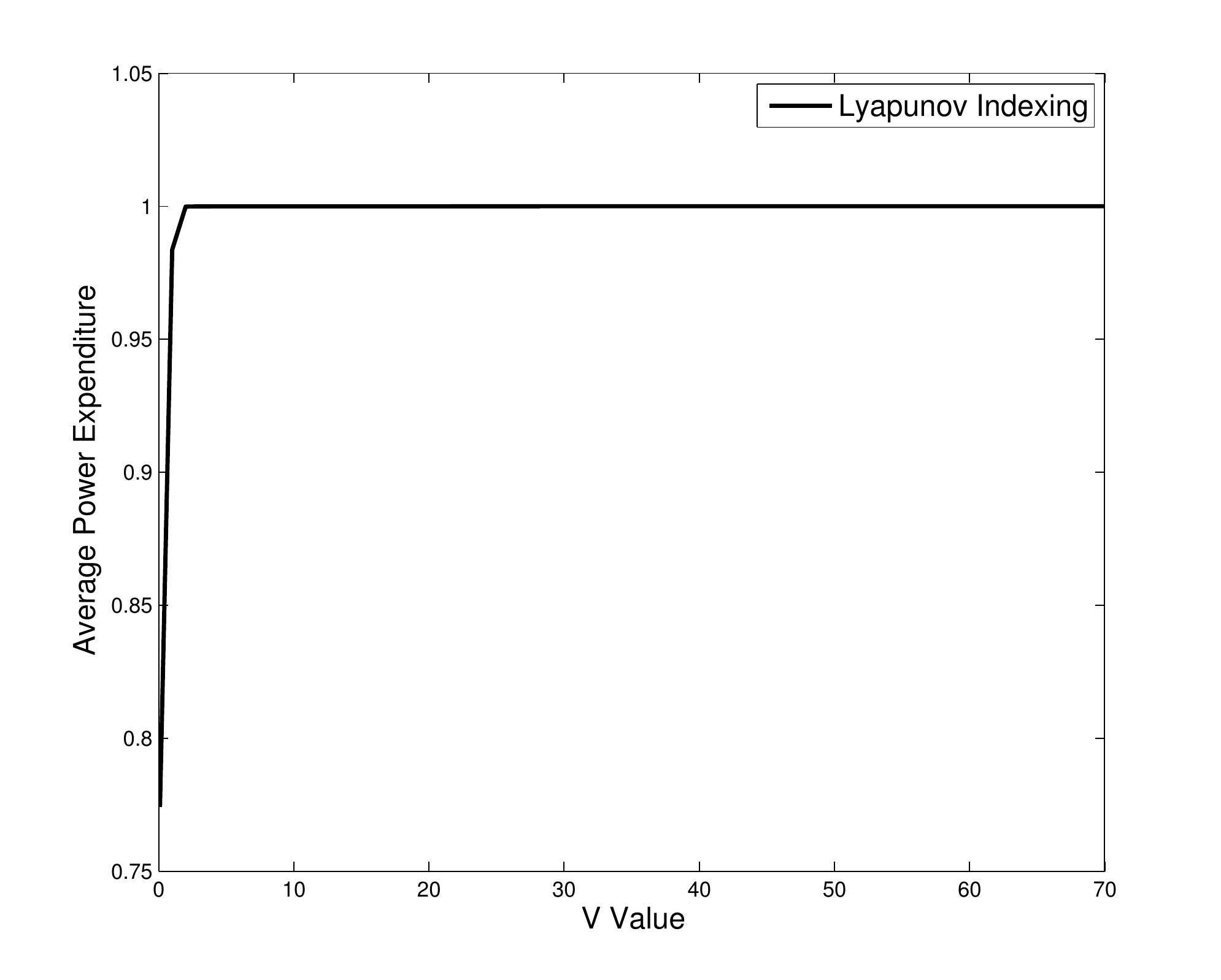} 
   \caption{The time average power consumption versus tradeoff parameter $V$.}
   \label{fig:Stupendous3}
\end{figure}

\begin{figure}[htbp]
   \centering
   \includegraphics[height=2.5in]{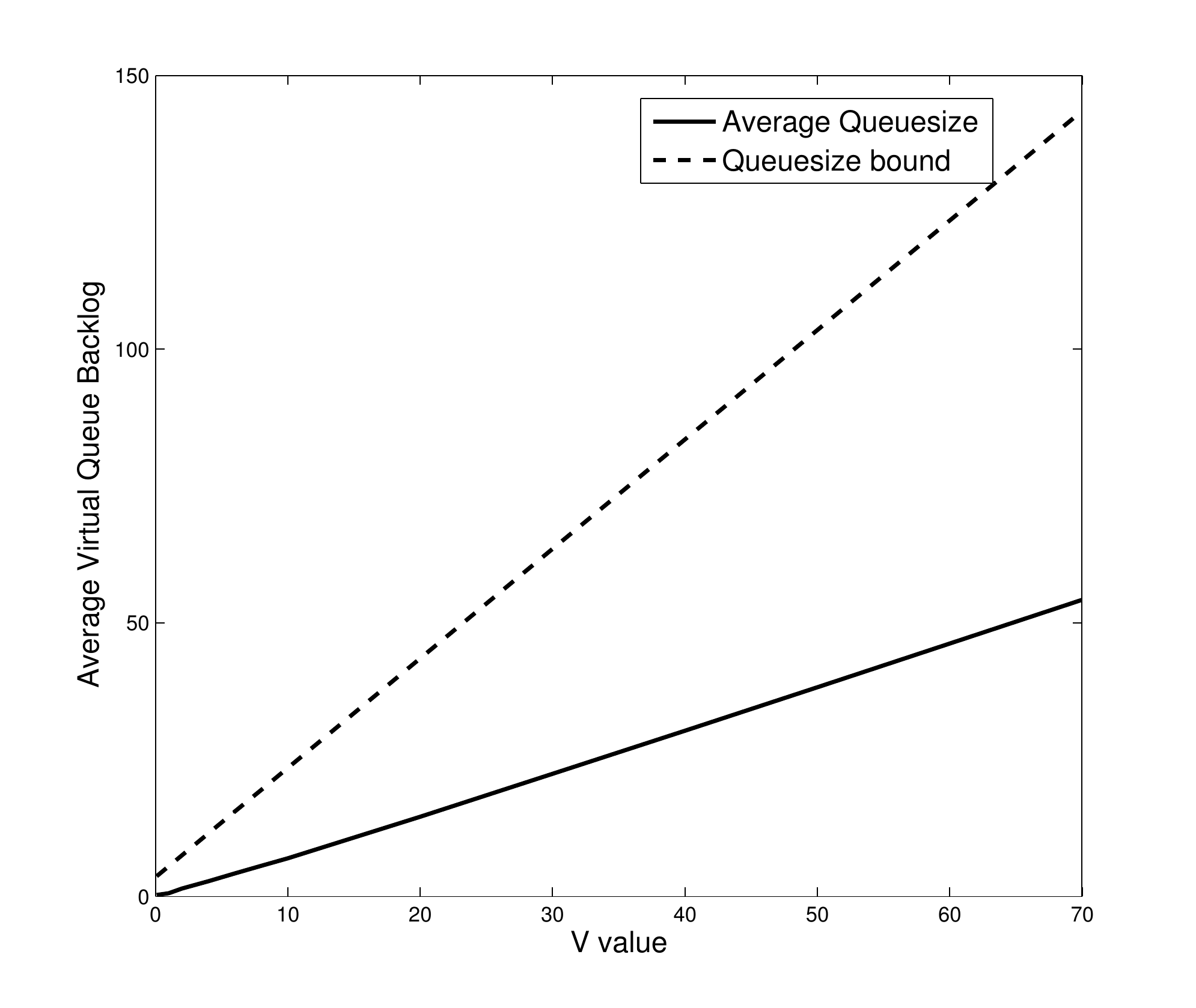} 
   \caption{Average virtual queue backlog versus tradeoff parameter $V$.}
   \label{fig:Stupendous4}
\end{figure}

In the second simulation, we explore the parameter space and demonstrate that in general the suboptimality gap of our algorithm is negligible. First, we define the relative error as the following:
\begin{equation}\label{e20}
    \textrm{relative error}=\frac{|OBJ-OPT|}{OPT}
\end{equation}
where $OBJ$ is the objective value after running 1 million slots of our algorithm and $OPT$ is the optimal value. We first explore the system parameters by letting $\lambda_{n}$'s and $\mu_n$'s take random numbers between 0 and 1, choosing $V=70$ and fixing the remaining
parameters the same as the last experiment. We conduct 1000 Monte-Carlo experiments and calculate the average relative error, which is \textbf{0.064\%}.

Next, we explore the control parameters by letting the $p_{n}(1)$ and $\phi_n(1)/\mu_n$ values take random numbers between 0 and 1, choosing $V=70$ and fixing the remaining parameters the same as the first simulation. The relative error is \textbf{0.077\%}. Both experiments show that the suboptimality gap is extremely small.

\subsection{Lyapunov indexing for multi-user downloading with non-memoryless file lengths}
In this part, we test the sensitivity of the algorithm to different file length distributions. In particular, the uniform distribution and the Poisson distribution are implemented respectively, while our algorithm still treats them as a geometric distribution with same mean. We then compare their throughputs with the geometric case.

We still use $N=3$, $M=1$ with action set $\mathcal{A}_n=\{0,1\}~\forall n$. For the uniform distribution case, the file lengths of the three subsystems are uniformly distributed between $[5,15]$, $[2,8]$ and $[1,5]$ packets, respectively, with integer packet numbers. For the Poisson distribution case, the Poisson parameters are set to ensure means of $10$, $5$ and $3$ packets, respectively.  We then keep the remaining conditions the same as the first simulation scenario in Section \ref{subsection:geometric-sim}.  In the algorithm we use $\phi_n(\alpha_n)$ functions defined
using parameters $\overline{B}_n = 1/\mu_n$ with $\mu_1 = 1/10$, $\mu_2 = 1/5$, $\mu_3 = 1/3$.  While the decisions are made using these values, the affect of these decisions incorporates the actual (non-memoryless) file sizes.
Fig. \ref{fig:Stupendous5}  shows the throughput-versus-$V$ relation for the two non-memoryless cases and the memoryless case with matched means.
Remarkably, the curves are almost indistinguishable. This illustrates that
the indexing algorithm is robust under different file length distributions.

\begin{figure}[htbp]
   \centering
   \includegraphics[height=2.5in]{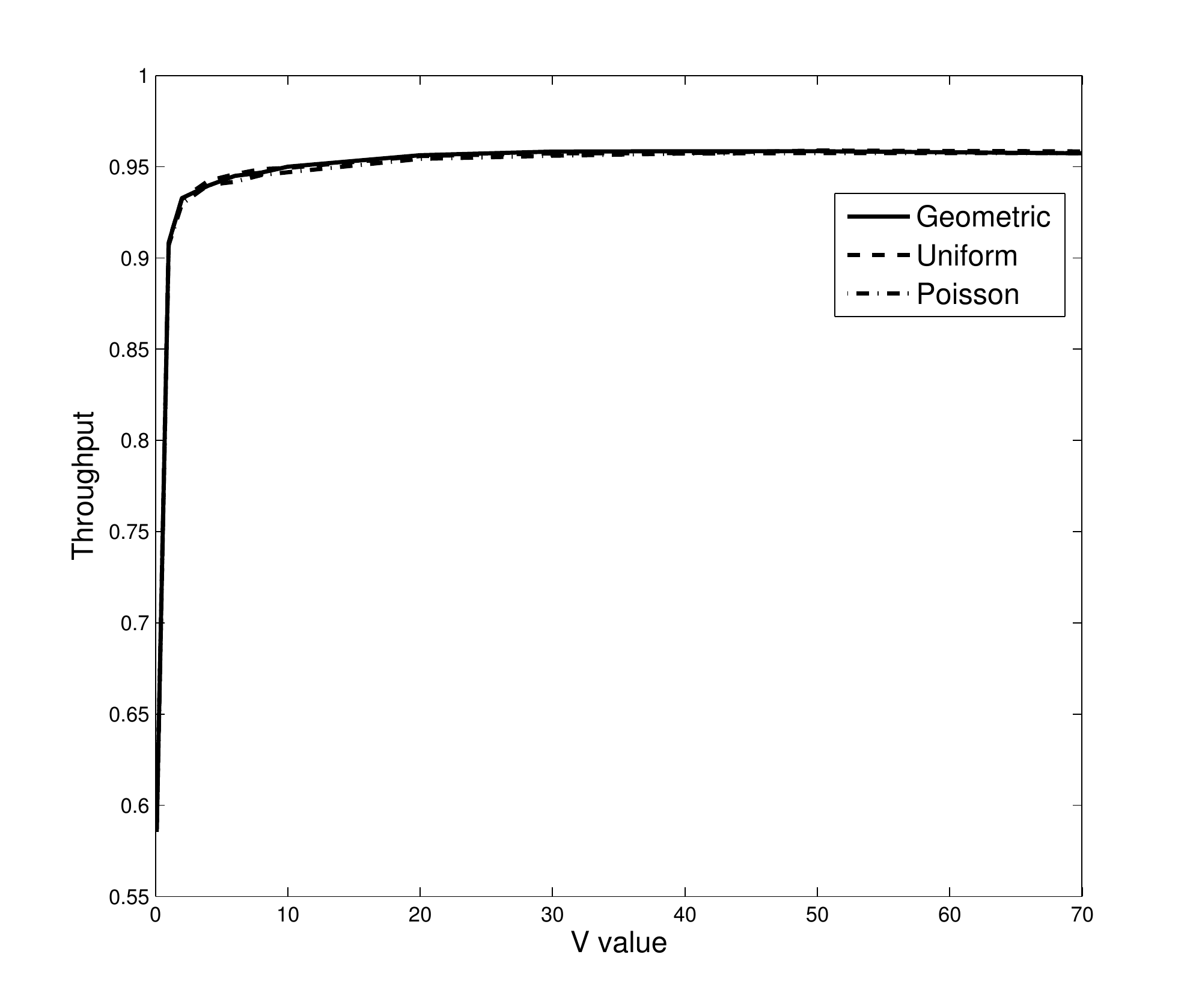} 
   \caption{Throughput versus tradeoff parameter $V$ under different file length distributions.}
   \label{fig:Stupendous5}
\end{figure}

\section{Conclusions}

We have investigated a file downloading system where the network delays affect the file arrival processes.  The single-user case was solved by a variable frame length Lyapunov optimization  method.  The technique was extended as a well-reasoned heuristic for the multi-user case.  Such heuristics are important because the problem is a multi-dimensional Markov decision problem with very high complexity.
The heuristic is simple, can be implemented in an online fashion, and was analytically shown to achieve the desired average power constraint. While we do not have a proof of throughput optimality for the multi-user case, simulations suggest that the algorithm is very close to optimal. Further, simulations suggest that non-memoryless file lengths can be accurately approximated by the algorithm.
These methods can likely be applied in more general situations of restless multi-armed bandit problems with constraints.

\bibliographystyle{unsrt}
\bibliography{bibliography/refs}
\end{document}